\documentclass[journal,12pt,draftclsnofoot,onecolumn]{IEEEtran}
\usepackage[]{url,epsfig,amsmath,amsthm}
\usepackage{algorithm}
\usepackage{algpseudocode}
\usepackage{eqparbox}
\usepackage{setspace}
\usepackage{framed}
\renewcommand{\algorithmiccomment}[1]{\hfill\eqparbox{COMMENT}{\% #1}}
\IEEEoverridecommandlockouts

\begin{document}
\title{Secondary Spectrum Auctions for Markets with Communication Constraints}

\author{\IEEEauthorblockN{Deepan
Palguna\thanks{Authors are with the
School of Electrical and Computer Engineering, Purdue University,
West Lafayette, Indiana 47907-2035. E-mail: \mbox{dpalguna@purdue.edu},
\{djlove,ipollak\}@ecn.purdue.edu. Tel: +1 765-496-1263, +1 765-496-6797,
+1 765-494-5916.}\thanks{
A shorter version of this paper excluding theoretical results in
Section~\ref{sec:matched_convergence}, the sections on strategic
users~\ref{sec:truthful_matched_auction}, multi-unit auctions~\ref{sec:vickrey_auction},
time-varying bids~\ref{sec:movingBidsScheme}, 
and the corresponding results in Section~\ref{sec:simulations} 
has been published in~\cite{PaLoPo13}.
}, David J. Love, and Ilya Pollak
}}
\maketitle
\pagestyle{empty}
\thispagestyle{empty}
\begin{abstract}
Auctions have been proposed as a way to provide
economic incentives for primary users to dynamically allocate
unused spectrum to other users in need of it. Previously proposed
schemes do not take into account the fact that the power
constraints of users might prevent them from transmitting their bid prices
to the auctioneer with high precision and that transmitted bid
prices must travel through a noisy channel. These schemes also
have very high overheads which cannot be accommodated in
wireless standards. We propose auction schemes where a central
clearing authority auctions spectrum to users who bid for it, while
taking into account quantization of prices, overheads in bid revelation,
and noise in the channel explicitly. Our schemes are closely
related to channel output feedback problems and, specifically,
 to the technique of posterior matching. We consider several
scenarios where the objective of the clearing authority is to award
spectrum to the bidders who value spectrum the most. We prove
theoretically that this objective is asymptotically attained by our
scheme when the bidders are non-strategic with constant bids.
We propose separate schemes to make strategic users reveal their
private values truthfully, to auction multiple sub-channels among
strategic users, and to track slowly time-varying bid prices. Our
simulations illustrate the optimality of our schemes for constant
bid prices, and also demonstrate the effectiveness of our tracking
algorithm for slowly time-varying bids.
\end{abstract}

\IEEEpeerreviewmaketitle
\section{Introduction}
The increasing interest in cognitive radio systems has led to the
development of the IEEE 802.22 and IEEE 802.16h standards~\cite{ieee80222,ieee80216}. These standards support some of the flexible and shared spectrum features of cognitive radios.
Both these cognitive radio standards have mechanisms for communication 
between base stations, which can enable sharing of unused spectrum among
unlicensed users who compete for it. 
In this setting, an economic incentive might be necessary for spectrum 
owners to be willing to allocate their unused spectrum to other users who are in need of it.
As a way of providing this incentive to the spectrum owner, secondary spectrum auctions 
have been proposed for dynamic spectrum allocation.

Spectrum auctions that account for interference constraints are proposed in~\cite{HuBeHo06},~\cite{JiZhZhLi09} and~\cite{GaBuCaZhSu07}. In~\cite{JiZhZhLi09} and~\cite{GaBuCaZhSu07}, computationally
efficient suboptimal schemes have been proposed to allocate multiple 
channels, with the objective of maximizing revenue. 
Online spectrum auctions, where users can bid for 
spectrum at any time when they need it, can be prone to manipulation, which might 
result in lower revenues for the auctioneer~\cite{DeZhAlZh11}. Spectrum sharing problems have been viewed from a game 
theoretic perspective in~\cite{EtPaTs07} where fairness and efficiency in 
spectrum sharing have been studied. There have also been other papers that consider 
pricing in secondary spectrum markets from a game theoretic
perspective~\cite{ScLa13,SaKa08,SeChBrCh08,ChHu10}. Auction based 
resource allocation has been studied for cooperative networks in~\cite{HuHaChPo08}.
While auctions are an economically appealing method of dynamic spectrum allocation, 
mechanisms have also been proposed in dynamic trading of spectrum contracts among 
primary and secondary users in~\cite{KaMuSaKaGu11}. 
Spectrum sharing based on contracts has also been considered 
for cooperative networks in~\cite{DuGaHu11}. There have been efforts in modeling 
secondary spectrum markets as double auction markets~\cite{XuJiLi10_1}, whose
structure closely resembles financial markets.
In~\cite{XuJiLi10_2}, the authors use a portfolio optimization approach to spectrum trading. 

The drawbacks of the schemes in the current literature
become clear when we look at the close connection between secondary 
spectrum auctions and user scheduling problems.
In a scheduling problem, a scheduler collects channel quality
information (CQI) from the users that it serves. Based on the CQI and
fairness considerations, the scheduler allocates time or frequency 
slots to users. It could, for example, assign the next channel use slot 
to the user with the best signal-to-noise ratio (SNR). 
But feedback from users in the form of instantaneous SNR would make 
such scheduling algorithms impractical as the number of users increases~\cite{HaTeAl05}. 
This is because the amount of power and bandwidth required for reliable 
feedback would be enormous. Therefore, a number papers have attempted to 
reduce feedback in user scheduling problems. One method to reduce feedback 
from users is for users to transmit quantized SNR information only if their SNR exceeds a particular threshold~\cite{GeAl04}. 
In~\cite{HaTeAl05} and~\cite{FlEdMo03}, the authors study the scenario where
the users transmit SNR information that is quantized using multiple levels.
It has been argued in~\cite{FlEdMo03} and~\cite{SaNo07} that increasing
the number of feedback bits results in diminishing improvements in throughput.

Two other challenges in user scheduling are due to latency and erroneous feedback bits. In~\cite{HaAlOi06}, the impact of latency on such schemes is analyzed, where a user could be allowed to transmit at a time slot based on outdated CQI. 
Similarly, delay is an important factor in auction design when feedback loads 
are high and the bids of the users change with time due to changing channel quality. 
The authors of~\cite{HaAlOi06} find that system performance degrades 
significantly with delay, even when the channel is slowly varying with time.
In~\cite{HaWi04}, the authors study the scheduling problem under
a practical scenario where there are errors in the feedback bits, and 
show that schemes that improve upon maximum SNR scheduling can be 
designed when 
the CQI is noisy. Literature surveys about limited feedback in wireless communications in
general, and with specific emphasis on adaptive transmission and scheduling
can be found in~\cite{LoHeLaGeRaAn08} and~\cite{ErOt07} respectively.

The challenges tackled in the user scheduling literature also affect the design of schemes for secondary spectrum auctions. In the context of both
secondary spectrum auctions and user scheduling, users have access only to finite-rate control channels, which is why they compete for spectrum in the
first place. Moreover, power constraints on mobile devices and delay requirements in the case of slowly time-varying bid prices put forth a strong
case for designing auction mechanisms where users are required to transmit a very small number of bits to the auctioneer for bid revelation and bid updates. 

In this paper, we propose several auction schemes under a scenario where there are two-way channels between the auctioneer (or clearing authority -- abbreviated as CA) and the users. Our schemes 
explicitly take into account the practical issues that arise
due to quantization requirements and noise. In such a set-up, the time period
of interest is divided into multiple rounds, where each round consists
of an update-and-allocate period and a spectrum use period. During an update-and-allocate period, each user can only transmit a small, fixed number of
bits to the CA through a noisy channel. This is because users are heavily
constrained by the available power and bandwidth.
The CA then chooses a winner for each spectrum unit under auction. Since the CA does not know the actual bids, it treats them as random variables and makes spectrum award decisions based on its estimates of these random variables. 
The winners can use the spectrum awarded to them during the next spectrum 
use period.
After the spectrum use period, the CA gets back control of the spectrum and begins the next round of allocations. A natural objective of the CA is 
to discover the true values of the users and allocate spectrum to users who value it the most as the number of auction rounds increases. 

The central contribution of our paper is a scheme which enables the 
CA to asymptotically achieve its objective of discovering the 
bid prices of the users and allocate spectrum to the highest bidder.
We moreover show that asymptotically, the CA's revenue can be made 
arbitrarily close to the highest bid.
At the beginning of an update period, the users send one bit each to 
the CA, which is a function of both the bid price and the feedback 
bits received from the CA. The bids are estimated by the CA 
from their posterior distributions, conditioned on the information
available to the CA till that round. The CA then sends two bits back to 
each user, one informing the user whether or not it is the winner for that
round, and the other informing the user about the CA's new bid 
estimate. The two bits from the CA to each user are assumed to be 
received without error due to the abundant communication resources at 
the CA's disposal. In the next round, the users reply back in the same 
fashion as before, and the process continues as long as the CA has a 
unit of spectrum to auction. 

This scheme is attractive firstly because we prove that it is asymptotically 
optimal. In other words, we prove that even under constraints
of very limited signaling and noisy transmission from the users to the CA,
our scheme guarantees that the CA asymptotically 
allocates spectrum to the highest bidder as the number 
of auction rounds increases. Secondly due to the small communication overheads, 
our scheme can be extended to handle other practical
issues like strategic bidders, auctioning multiple units of spectrum and accounting for time varying bid prices.
Our method is closely related to the technique of posterior matching \cite{ShFe07}, \cite{Ho63}, due to which we 
call it {\em matched auctioning}.
Since the CA is typically a base station that can transmit using a large amount of power, we assume that the channel from the CA to each user is noiseless (whereas the user-to-CA channels are noisy).  This assumption is critical to the optimality of our proposed scheme. In the presence of noise in the CA-to-user channels, noise would accumulate with each round and this case warrants
further investigation.
The organization and the main contributions of our paper are as follows:

{{\bf Auction design under practical communication constraints:}} We model quantization and noise for the first time in the context of 
secondary spectrum auctions and devise schemes for auctions under practical constraints. In the next section, we describe the system set-up and provide an outline
for single-unit multi-round auctions. 

{\bf{Unmatched auctioning:}} To study the behavior of the schemes in the
existing literature under communication constraints, we first 
propose a scheme to auction one unit of spectrum, where the users do 
not utilize the feedback bits from the CA to decide their future transmissions.
This motivating example is suboptimal since it does not provide any allocation
guarantees, and is described in Section~\ref{sec:unmatched}.

{\bf{Matched auctioning:}} Our central scheme to auction one spectrum unit 
among non-strategic users is described in Section~\ref{sec:matched}. 
We prove that this scheme is asymptotically optimal in the sense of
getting arbitrarily close to 
maximizing the CA's revenue and allocating spectrum to the highest 
bidder as the number of auction rounds increases. Following this, we 
propose three separate extensions accounting for other practical considerations. These extensions illustrate the importance of schemes with 
low communication overheads and the scalability of matched auctioning.

{\bf{Quantized single-unit auctions with strategic users:}}
In Section~\ref{sec:truthful_matched_auction} we propose a 
single unit auction scheme called {\em truthful matched auctioning} 
that can handle strategic users. These are non-cooperating and
rational users that attempt to maximize their payoff.
Under truthful matched auctioning, truthful bid revelation 
is weakly dominant as the number of update rounds increases. This result 
is suggested from our simulation results.

{\bf{Quantized Vickrey auctions:}}
As the second extension, we propose a scheme to simultaneously auction 
multiple units of spectrum among strategic users. Simulations of this scheme also show that truthful bid revelation is a weakly dominant strategy as the number of rounds increases.
Quantized Vickrey auctioning can be viewed as a generalization of truthful
matched auctioning, and is described in Section~\ref{sec:vickrey_auction}.

{\bf{Matched auctioning with slowly time-varying bids:}}
Constant bid prices can be a strong assumption for wireless systems. For example, a user could be a mobile device that wants to vary its bid due to changing channel conditions. 
Simulations of the scheme that we propose in Section~\ref{sec:movingBidsScheme} for this scenario show that our tracking method gives revenues 
close to the optimal revenue, and outperforms matched auctioning for a wide range of parameters that govern bid price dynamics.
Our simulation results and the conclusion are presented in Sections~\ref{sec:simulations} and~\ref{sec:conclusion} respectively.

\section{System set-up and single-unit auction scheme outline}
\label{sec:sysset}
Consider the scenario of $N$ users bidding for one unit of spectrum
that is being auctioned by a central clearing authority --- abbreviated as CA.  
The CA is a base-station and the users could be wireless devices in a particular cell, or even other base stations.
\begin{figure}[htb]
\begin{minipage}[b]{1.0\linewidth}
  \centering
  \centerline{\epsfig{figure=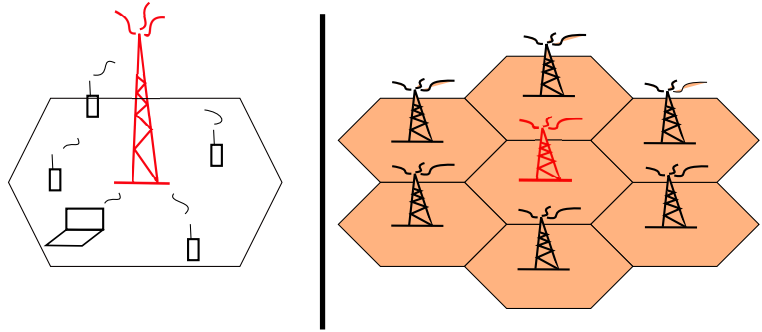,width=10cm}}
\end{minipage}
\caption{Scenarios with different kinds of users and CAs. CAs shown in red.}
\label{fig:scenario}
\end{figure}
\subsection{Definitions and preliminaries}
\label{sec:auction_definitions}
Secondary spectrum auctions are modeled as {\em private value} auctions,
in which the $i^{\text{th}}$ user attaches a value ($v_i$) to the object under auction.
These values are only known to the respective users.
Prior to receiving any information from the users, the auctioneer models these values as i.i.d. random variables. 
In spectrum auctions, this distribution models channel conditions, user requirements, and other factors which would affect the value of one spectrum
unit.
The strategy of the $i^{\text{th}}$ user is a mapping from its true value $v_i$
into a bid price i.e., $s_i(v_i) = b_i$. In a {\em standard} auction, the $i^{\text{th}}$ user will win the auction if it has the highest bid. The auctioneer then charges an ask price equal to $a$, giving the winner a payoff equal to $v_i-a$ and zero payoff for the others. 
A strategic user is one which behaves so as to maximize its payoff.
If a user is non-strategic, then its strategy is the identity function. A non-strategic user is also assumed to always be truthful and to adhere to the auction rules, even if deviating from the rules would give it higher payoffs.

A natural choice for the ask price that the winner gets charged, is the winning bid itself. Such an auction is called a {\em first price auction}. 
Bidding one's own value in a first price auction would only guarantee a payoff of zero. Therefore, in general first price auctions,
strategic bidders will not bid their true private values.
On the other hand, for a standard auction where the ask price is equal to the 
second highest bid, the strategy $s_i(v_i) = v_i$ is a {\em weakly dominant strategy}
for each user. This means that irrespective of what other users do, a user would not receive a better payoff if it did not bid truthfully. This is a good property for an auction to 
have since each user knows what to do irrespective of what other users do.
The definitions  and results introduced here are standard in the auctions literature~\cite{Kr02}.
In the next subsection, we describe the outline of our single-unit 
auction schemes. The set-up of the multi-unit auction scheme in Section~\ref{sec:vickrey_auction} is very similar to the set-up
of truthful matched auctioning. So we explain the set-up for
quantized Vickrey auctions in the corresponding
section for ease of description and clarity in conveying the main ideas.
\subsection{Single-unit auction scheme outline}
\label{sec:scheme_outline}
Depending on channel conditions, individual requirements and their strategies,
the users fix their bids as $b_1, \cdots, b_N$,
which are all assumed to lie in the interval $[0,1]$. 
We assume that there is a two-way 
channel between each user and the CA, and there is no interference between these channels. Since the CA is typically a base station with high transmit power and unutilized bandwidth that is dedicated to control, we assume that the CA-to-user channels are noiseless. The CA can award spectrum to the highest bidder of each round in one shot if the users 
could send their bids to the CA with infinite precision. But in our set-up, we consider quantization and noise constraints explicitly. This results
in the CA refining its estimate of the highest bid from round to round. 
Each round is divided into two disjoint intervals: an update-and-allocate period and a spectrum use period. This is depicted in Fig.~\ref{fig:timing_diagram}.
\begin{figure}[htb]
\begin{minipage}[b]{1.0\linewidth}
  \centering
{\epsfig{figure=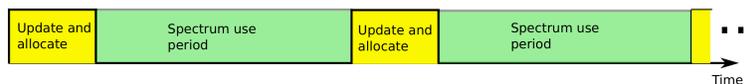,width=10cm}}
\vspace*{-0.25cm}
\end{minipage}
\caption{The CA has control of the spectrum during update-and-allocate period. Users only have control channels to communicate with the CA. The CA updates its bid estimates and decides to allot spectrum to one of the users for the following spectrum use period. Based on the updated bid price estimates, the CA updates the ask price it wants to charge the winner.
The CA gets back control of spectrum following the spectrum use
period, and the process of update-and-allocate continues as long as the CA is willing to auction spectrum.}
\label{fig:timing_diagram}
\end{figure}

During an update-and-allocate period, the users have severely constrained channels connecting them to the CA. So each of those periods is meant to refine the CA's estimate of the user's bids while using very limited signaling.
We now list the steps that take place in the $t^{\text{th}}$ update-and-allocate period for single-unit auctions.
\begin{itemize}
\itemsep0em
\item At the start of round $t$, user $i$ is allowed to send only one bit to the CA, denoted by $x_{it}$. In general, $x_{it}$ is a function of $b_i$ and all the other
information available to the user until round $t$. Due to noise in the user-to-CA channel, $x_{it}$ is received by the CA as $y_{it}$.
\item Since the CA does not know the bids, it models them as independent continuous
random variables $\{B_i\}_{i=1}^N$, each uniform over $[0,1]$.
\item Using $(y_{1t},\cdots,y_{Nt})$, and all the bits received during the previous rounds, the CA estimates each bid and awards spectrum for the corresponding spectrum use period to the user whose bid price estimate is the highest. Ties are broken arbitrarily.
\item A spectrum ask price $a_t$ is fixed by the CA based on its updated bid
estimates.
\item The CA then sends the first feedback bit $u_{it}$ to each user $i$, which is to inform the user whether it
was awarded spectrum for the following spectrum use period or not. This is given by
\begin{equation}
u_{it} = \
\begin{cases}
1 & \text{ if user $i$ won round $t$}
\\
0 & \text{ otherwise.}
\end{cases}
\label{eqn:uit}
\end{equation}
We can write this as $u_{it}=I_{\text{user $i$ won round $t$}}$, where $I_\mathcal{A}$ is the indicator function of event $\mathcal{A}$.
\item The second feedback bit sent from the CA to the user is $z_{it}=y_{it}$. This bit is sent so that all the users can perform the same updates as the CA and compute the
CA's new bid estimate. 
When users are strategic, the CA has to send a third bit, $\tilde{y}_{it}$, to enable the users to compute the ask price.
Bits sent by the CA are correctly received by 
the users due to noiseless feedback. 
\item When users are not strategic, the winner 
has the option to reject spectrum and pay nothing if the ask price is larger than its bid. If the winner exercises this option, then the CA's revenue during round $t$ would be zero, and spectrum will be unused in the following spectrum use period. Otherwise the winner will choose to accept spectrum, and the CA would get a revenue equal to the ask price $a_t$. In truthful matched auctioning though, 
winners are not allowed to reject spectrum since they are strategic.
Therefore, they always use spectrum, giving a revenue of $a_t$ to the 
CA during each round.
\item Allowing winners to reject spectrum when they are not strategic is beneficial to the winners. Although it could reduce the revenue of the CA during the
initial update rounds, we prove that the revenue converges to
a price close to the maximum
bid price in probability as the number of update rounds increases.
Not allowing winners to reject spectrum in truthful
matched auctioning tackles the problem of strategic
bidders at the cost of winners having to sometimes pay a price
larger than their bid. But we will see using simulations that as
the number of update rounds increases, the probability of winners paying an amount smaller than their bid converges to one.
\item Subsequent to the corresponding spectrum use period, the CA gets back control of the spectrum and the
users will send $x_{i,t+1}$, just like in the previous round. This procedure continues as long as the
CA has a unit of spectrum to auction. 
The steps in one update-and-allocate period for single-unit auctions are illustrated in Fig.~\ref{fig:auctionOutline}. In this section, we have left out the exact equations that are used by each scheme to compute $x_{it}$ and $a_t$.
These will be addressed in the corresponding sections. We will also
address the computation of $\tilde{y}_{it}$ in 
Section~\ref{sec:truthful_matched_auction} on truthful matched auctions.
\item In practice, the final payments can be made to the CA at the end of
the auction. When users are strategic, the CA has to remember only the
winner information and collect the corresponding ask price from the 
winners of each round. When users are non-strategic, we additionally assume
that the winners remember their usage information and pay the CA
truthfully.
\end{itemize}
\begin{figure}[htb]
\begin{minipage}[b]{1.0\linewidth}
  \centering
  \centerline{\epsfig{figure=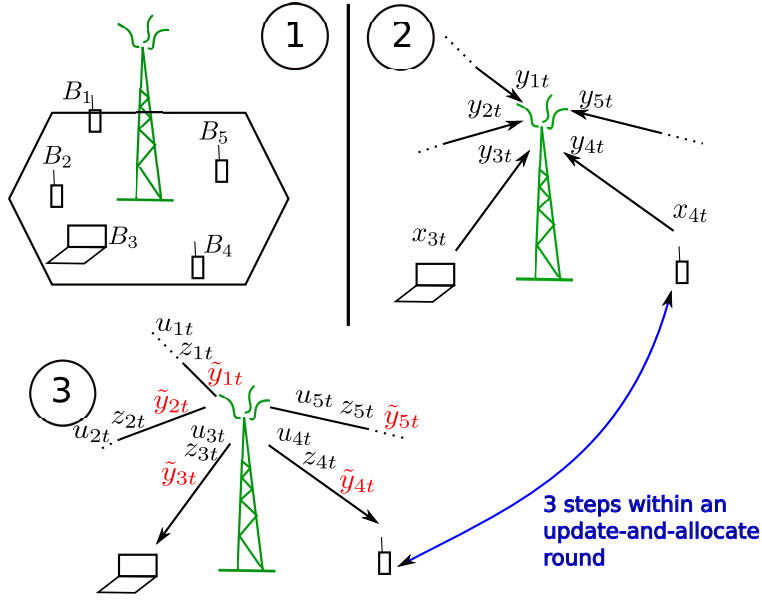,width=10cm}}
\end{minipage}
\caption{Update period outline for single-unit auctions: 
Users (the five devices in black) 
send 1 bit each to the CA (the base station in green). The bit can be a function of the user's bid and all the 
other information available to the user. The actual form of $x_{it}$ depends on the 
specific scheme. CA updates the bid and ask price estimates. It then sends the received bit and winner information back to the users.
The third bit $\tilde{y}_{it}$, highlighted in red, is only required
to convey the ask price to strategic users in truthful matched auctioning.
This bit is not needed in single-unit auctions with non-strategic bidders.}
\label{fig:auctionOutline}
\end{figure}

While the CA-to-user channels are assumed to be noiseless, 
we model the user-to-CA channels as non-interfering binary symmetric channels (BSC). 
If the input to a BSC is 1, then it will be received erroneously 
as 0 with probability $p$. Similarly, an input of 0 will be received 
as 1 with probability $p$. 
BSC$_p$ is used to denote a BSC with cross-over probability $p$. 
\subsection{Accounting and payment method}
\label{sec:accounting}
There needs to be a mechanism for the CA to get paid at the end of the time horizon 
or at the time when it decides to stop auctioning spectrum. This can be taken 
care of by the CA storing winner information and the users storing usage
information. For example, if the CA allocates spectrum to users $3,4,4,4,4,5$ 
during rounds $1-6$, then the CA stores this information. If user 4 accepts during rounds 3, 5 and rejects during rounds 2, 4, then it stores the sequence $0,1,0,1$. 
At the time of payment, which is the end of all update-and-allocate rounds, the user sends this sequence reliably to the CA, for instance, using an open-loop channel code. The CA then charges the user based on the ask prices during the rounds in which spectrum was accepted. Here we assume that users are non-strategic, an assumption
that we use throughout the paper and relax only in 
Sections~\ref{sec:truthful_matched_auction},~\ref{sec:vickrey_auction} and~\ref{sec:truth_results}, where we deal with quantized auction schemes for strategic users. When users are strategic,
winners are not allowed to reject spectrum and therefore pay the ask
price during each round. So it is enough if the CA remembers the
winner information.
\section{Quantized auction example: Unmatched auction}
\label{sec:unmatched}
As a motivating example for quantized auctions and as a way to
compare our schemes with the schemes in the literature that assume 
noiseless user-to-CA transmission, we now propose a quantized 
version of single unit auctions with non-strategic users.
In this simple quantized auction, user $i$ sends one bit per 
update-and-allocate period without using the second feedback bit ($z_{it}$) 
it has received from the CA. The sequence of bits is obtained by converting 
the bid into its binary equivalent. For example, if user $i$'s bid is 0.76, 
then the binary equivalent is $0.1100001\cdots$. The sequence of bits 
obtained from the binary equivalent forms the sequence 
$\{x_{i1}, x_{i2},\cdots\}$. Therefore, the first three 
transmissions from user $i$ would be $1,1,0$. The CA's estimate of the 
bid prices at each update round is obtained by converting the binary 
sequence it has received from each user until that point, into the 
corresponding decimal fraction. 
For example, if the transmitted sequence $1,1,0$ is received with an 
erroneous third bit as $1,1,1$, then the CA's estimate of the bid after 
the third reception would be 
$1 \times 2^{-1}+1\times 2^{-2}+1\times 2^{-3} = 0.875$. 
In any round $t$, the CA awards spectrum to the user with the highest 
estimate and sets this estimate as the ask price for that round. It then 
sends the two feedback bits $u_{it}$ and $z_{it}$ to inform the users about 
the result of the $t^{\text{th}}$ update-and-allocate round and about 
the CA's estimate of its bid price. Since the users know $y_{it}$, the 
winner also knows the ask price set by the CA. Based on this information, 
the winner chooses either to accept or reject spectrum for the $t^{\text{th}}$
spectrum use period. The winners can be given this choice since we
assume that the users are non-strategic.

When there is noise in the user-to-CA channel, this scheme would, on the average, result in sub-optimal allocations even after many update rounds. 
In other words, on the average, the CA is not guaranteed to allocate spectrum 
to the highest bidder, as illustrated in Section~\ref{sec:simconv}. 
In order to overcome this, one approach would be to view this as an 
open-loop channel coding problem and use sophisticated channel coding 
schemes to quantize the bids so that the CA's estimate is arbitrarily close to the actual bid. Although there exist techniques by which users can encode 
their bids to be received by the CA with arbitrarily low error probabilities, such techniques would require large block length open-loop codes that can cause severe overhead, resulting in the CA taking many rounds to converge to 
the optimal allocation. Such {\em{open-loop}} schemes also do not scale well with multi-unit auctions and with time-varying bids. 
\section{Matched auction scheme}
\label{sec:matched}
For single-unit auctions with non-strategic users and constant bid
prices, an alternative approach to unmatched auctioning is to 
exploit the noiseless feedback from the CA to the users.
In this section, we devise such a scheme, where
as the number of rounds increases, the probability that the CA allocates spectrum to the highest bidder approaches one and 
the CA's revenue converges to a price that is close to the 
maximum bid price.
\\
{\bf Posterior matching and channel output feedback problems:}
Since this scheme is closely related to the iterative scheme in 
Horstein's paper~\cite{Ho63}, we discuss it briefly.
Horstein's scheme is a specific case of the more general framework of posterior matching.
It achieves the point-to-point capacity for a BSC with noiseless feedback.
In this scheme, the transmitter represents the sequence it wants to transmit using
a {\em message point}. The receiver knows the prior distribution of the message point,
which is continuous and uniform over the interval $[0,1]$. 
The transmitter knows that the receiver's prior model for the distribution of the message point
is uniform over $[0,1]$.
Both the transmitter and the receiver maintain
and update the posterior distribution of the message point conditioned on all
the bits observed at the receiver.
In each round, the transmitter tells the receiver whether or not the message point
is below the posterior median. The receiver uses this bit to update its posterior
distribution and sends the same bit back to the transmitter. The feedback bit is
received error-free because of the noiseless feedback. Therefore, the transmitter
can perform the same posterior update as the receiver. In the next round, the transmitter
sends one bit according to the same rule as in the previous round. As the number of rounds
increases, the receiver becomes more and more confident about its estimate of the message
point, and the posterior
cumulative distribution of the message point converges to a unit step at the actual message point.

In our scheme, the $N$ bid prices act as $N$ message points and are represented as 
random variables $\{B_i\}_{i=1}^N$. The users act as transmitters, while the CA acts 
as the receiver and maintains posterior distributions for each of these bid prices. 
At the beginning of round $t$, the users inform the CA whether their bids are at least as
large as the posterior median\footnote{The median $m$ of a continuous random variable $X$ is defined by
${\bf P}\{X \leq m\} = 1/2$. The median $m$ of a discrete random variable $X$ is defined
by ${\bf P}\{X \leq m\} \geq 1/2$ and ${\bf P}\{X \geq m\} \geq 1/2$. To make the median
unique in the discrete case, we pick the smallest $m$ such that it also has non-zero probability mass.}
$M_{it}$
at the beginning of that round:
\begin{equation}
x_{it} = \
\begin{cases}
1 & \text{ if $b_i \geq M_{it}$}
\\
0 & \text{ otherwise.}
\end{cases}
\label{eqn:xit}
\end{equation}
Each of these bits passes through a BSC$_p$ and is received by the CA as $y_{it}$.
Depending on the information it receives from the users, the CA updates its distribution of
$\{B_{i}\}_{i=1}^N$ and its estimate of each bid, which is set to be the updated posterior
median $M_{i,t+1}$. The winner for the $t^{\text{th}}$ round is picked using
${\arg \max}_i\{M_{i,t+1}\}$, and the ask price is set to be
$a_t = {\max}_i  \{M_{i,t+1}\} - h$, 
where $h$ is a small positive number. The number $h$ specifies how much less than the highest bid the CA will asymptotically get in revenue. The larger it is, the better it insures that the winner does not reject spectrum as 
$t$ increases, resulting in a faster convergence of the CA's revenue.
The role of $h$ would become clear in the proof of Proposition~\ref{lemma:lemma2}.
The CA then sends $u_{it}$, identical to the definition in (\ref{eqn:uit}) and $z_{it}=y_{it}$
so that all the users can perform the same update as the CA and compute the updated posterior
distribution. The winner then decides whether or not to use spectrum during the corresponding 
spectrum use period, based on its bid price and the new ask price. 

After the following spectrum use period, the users reply back to the CA in 
the same fashion as in the previous round using (\ref{eqn:xit}) and 
the process continues as long as the CA is willing to sell spectrum. 
These steps are illustrated in Fig.~\ref{fig:auction2}.
\begin{figure}[htb]
\begin{minipage}[b]{1.0\linewidth}
  \centering
  \centerline{\epsfig{figure=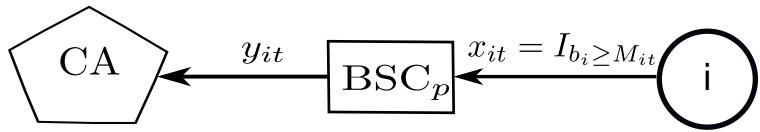,width=10cm}}
\end{minipage}
\caption{User$_i$-to-CA step at the beginning of round $t$.}
\label{fig:auction1}
\end{figure}
\begin{figure}[htb]
\begin{minipage}[b]{1.0\linewidth}
  \centering
  \centerline{\epsfig{figure=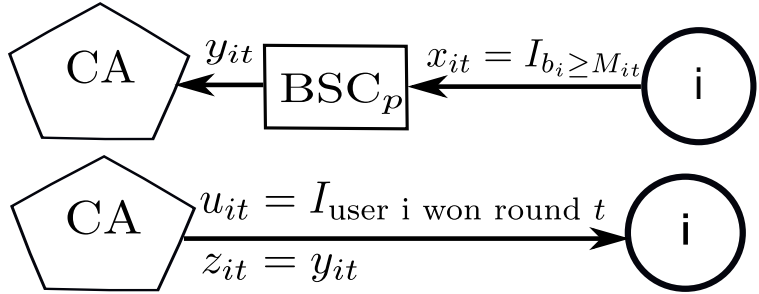,width=10cm}}
\vspace*{-0.25cm}
\end{minipage}
\caption{This figure shows the User$_i$-to-CA step at the beginning of 
the $t^{\text{th}}$ update round and CA-to-user$_i$ step at the end of 
the $t^{\text{th}}$ update round.}
\label{fig:auction2}
\end{figure}
In all our auction schemes, the payment to the CA is settled at the end of
all update-and-allocate rounds as outlined in Section~\ref{sec:scheme_outline}.
In the next two sub-sections, we provide a detailed description of
our algorithm and the posterior update step.
\begin{figure}
\fbox{
 \addtolength{\linewidth}{-2\fboxsep}%
 \addtolength{\linewidth}{-2\fboxrule}%
\begin{minipage}{\linewidth}
\begin{spacing}{1}
\begin{algorithmic}
\State $t \gets 1$ \% $t$ denotes the current update-and-allocate round
\While{(CA has a unit of spectrum to auction)}
 \For{$(i=1;i\leq N;i++)$}
   \State user $i$ executes \textsc{OneRoundUser}$(i,t)$ 
 \EndFor
 \State CA executes \textsc{OneRoundCA}$(t)$ 
 \State $t \gets t+1$   
\EndWhile
\end{algorithmic}
\end{spacing}
\end{minipage}
} 
\caption{Overall flow of the auction procedure.}
\label{fig:overallAlgo}
\end{figure}

\subsection{Algorithm description and pseudocode}
\label{sec:matched_algo_pcode}
The pseudocode for the matched auction scheme is in Fig.~\ref{fig:overallAlgo}, which
contains calls to functions \textsc{OneRoundUser} and \textsc{OneRoundCA},
both of which are shown in Fig.~\ref{fig:oneround}.
Function \textsc{OneRoundUser} is executed
by each user $i$ in every update-and-allocate round $t$, to mimic the CA's posterior median calculations
and determine $x_{it}$ according to (\ref{eqn:xit}).
Moreover, if spectrum has been awarded to the user,
this function decides whether or not to use the spectrum based 
on whether or not its bid is above $m_{it} - h$. 
Note that the posterior distribution of user $i$'s bid and the 
particular realization
of its median calculated by the CA,
are denoted by $F_{it}$ and $m_{it}$, respectively; whereas their replicas calculated and maintained
by user $i$ are denoted, respectively, $F^u_{it}$ and $m^u_{it}$.

Function \textsc{OneRoundCA} is executed by the CA in every update-and-allocate round $t$, in order to
update the posterior distribution of each user's bid, determine the auction winner, and calculate the two bits to be sent back to each user.
Before getting any information from the users, the CA's initial estimate
of the $i^{\text{th}}$ user's bid is the prior median of $B_{i}$, i.e., $1/2$.
Both functions \textsc{OneRoundUser} and \textsc{OneRoundCA} use the function \textsc{updateDistribution}
shown in Fig.~\ref{fig:updateDist}.  This function computes the update of the posterior
distribution of user $i$'s bid, based on the latest bit received by the CA from user $i$.
The update equations used by this function are given in the next subsection.
\begin{figure}[!htb]
\begin{minipage}{1\linewidth}
\begin{spacing}{1}
\begin{framed}
\begin{algorithmic}
{
\Function{\textsc{OneRoundUser}}{$i,t$}
 \If{$t<=2$}
   \State \% Initialize the distribution to uniform
   \State $F^u_{i,t-2} \gets \mbox{uniform}[0,1]$
 \Else
   \State retrieve previously stored $F^u_{i,t-2}$
 \EndIf
 
 \State $m^u_{it} \leftarrow \text{Median}(F^u_{i,t-2})$
 \If{$t>1$}
  \State receive bits $u_{i,t-1}$ and $z_{i,t-1}$ from the CA
  \State {\footnotesize{$F^u_{i,t-1} \leftarrow \textsc{{\scriptsize{updateDistribution}}}(F^u_{i,t-2},z_{i,t-1})$}}
  \State store $F^u_{i,t-1}$
  \State $m^u_{it} \leftarrow \text{Median}(F^u_{i,t-1})$
  \If{$(u_{it}=1$ \& $m^u_{it}-h \leq b_i)$}
   \State use spectrum
   \State store $t$
   \State store the amount $m^u_{it}-h$ to be paid
  \EndIf
 \EndIf
 \State $x_{it} \leftarrow I_{b_i \geq m^u_{it}}$
 \State send $x_{it}$ to the CA
\EndFunction
}
\end{algorithmic}
\end{framed}
\end{spacing}
\end{minipage}
\caption{Algorithm at the user $i$ during the $t^{\text{th}}$ round.}
\label{fig:user2ca}
\end{figure}
\begin{figure}
\begin{minipage}{1\linewidth}
\begin{spacing}{1}
\begin{framed}
\begin{algorithmic}
{
\Function{\textsc{OneRoundCA}}{$t$}

\For {$(i=1$; $i\leq N$; $i++)$} 
  \If{$t==1$}
   \State \% Initialize the distribution to uniform
   \State $F_{i,t-1} \gets \mbox{uniform}[0,1]$
  \Else
   \State retrieve previously stored $F_{i,t-1}$
  \EndIf
  \State receive $y_{it}$ from user $i$
  \State $F_{it} \gets \textsc{\footnotesize{updateDistribution}}(F_{i,t-1},y_{it})$
  \State store $F_{it}$
  \State $m_{i,t+1} \gets \text{Median}(F_{it})$
\EndFor
\State Winner $\gets {\arg \max}_i  \{m_{i,t+1}\}$
\State store Winner
\For {$(i=1$; $i\leq N$; $i++)$}
  \State $u_{it} \gets I_{i = \text{Winner}}$
  \State $z_{it} \gets y_{it}$
  \State send $u_{it}$ and $z_{it}$ to user $i$
\EndFor
\EndFunction
}
\end{algorithmic}
\end{framed}
\end{spacing}
\end{minipage}
\caption{Pseudo-code of the algorithm at both the user and CA during the $t^{\text{th}}$ round.}
\label{fig:oneround}
\end{figure}

\subsection{Posterior update step}
\label{sec:matched_posteriorUpdates}
All the bits sent by user $i$ until the beginning of round $t$ are received by the CA 
as $(y_{i1},\cdots,y_{it})={\bf Y}_{i,t}$. We denote the posterior distribution
of $B_i$, conditioned on ${\bf Y}_{i,t}$, using
\begin{equation}
F_{it}(b)=\mathbf{P}\{B_i \leq b | {\bf Y}_{i,t} \}.
\label{eqn:pdist}
\end{equation}
As stated earlier, we denote the posterior median of $B_i$ at the beginning of round
$t$ using $M_{it}$ defined as
\[
M_{it} = \text{Median}(B_i | {\bf Y}_{i,t-1}).
\]
$M_{it}$ is a random variable since it is a function of the random vector
${\bf Y}_{i,t-1}$.  We denote by $m_{it}$ the particular
realization of $M_{it}$ computed by the CA from a specific observation
of ${\bf Y}_{i,t-1}$.  We assume that each user-to-CA channel is a BSC$_p$,
that the channel noise is temporally independent, and that the channel
noise is independent of the bid prices.

Our model and convergence analysis in Section~\ref{sec:matched_convergence}
are based on assuming that prices are continuous random variables.
However, our simulations in Section~\ref{sec:simulations} reflect a more practical
scenario where prices are discrete.  We therefore derive posterior
update equations that are applicable to both the continuous and
the discrete cases.  Specifically, we assume that the discretization
interval for the bid prices is $\Delta$.  In the continuous case, $\Delta=0$.
In the discrete case, $\Delta > 0$, and the bids are integer
multiples of $\Delta$, with probability 1.  We define
\[
m'_{it} = m_{it} - \Delta.
\]
Note that, in the continuous case, $m'_{it} = m_{it}$.  In the discrete
case, $m'_{it}$ is equal to the highest possible bid price which is
below $m_{it}$.

\newtheorem{proposition}{Proposition}
\begin{proposition}[]
\label{prop:prop1}

The equations needed to calculate the posterior distribution $F_{it}$ from
$F_{i,t-1}$ and the bit $y_{it}$ received from user $i$ in round $t$,
are as follows.
\\
Case 1: $b < m_{it}$ and $y_{it}=1$
\\
\begin{equation}
F_{it}(b)= \dfrac{p F_{i,t-1}(b)}{1 - p - (1-2p) F_{i,t-1}(m'_{it})}.
\label{eqn:update11}
\end{equation}
\\
Case 2: $b < m_{it}$ and $y_{it}=0$
\\
\begin{equation}
F_{it}(b)= \dfrac{(1-p)F_{i,t-1}(b)}
                 { p + (1-2p) F_{i,t-1}(m'_{it})}.
\label{eqn:update01}
\end{equation}
\\
Case 3: $b \geq m_{it}$ and $y_{it}=1$ 
\begin{equation}
F_{it}(b)= \dfrac{(1-p) F_{i,t-1}(b) + (2p-1) F_{i,t-1}(m'_{it})}
                     {1-p + (2p-1) F_{i,t-1}(m'_{it})}.
\label{eqn:update12}
\end{equation}
\\
Case 4: $b \geq m_{it}$ and $y_{it}=0$
\begin{equation}
F_{it}(b)= \dfrac{p F_{i,t-1}(b) + (1-2p) F_{i,t-1}(m'_{it})}
                {p + (1-2p) F_{i,t-1}(m'_{it})}.
\label{eqn:update02}
\end{equation}
\end{proposition}
{\bf Proof} is in Appendix~\ref{sec:derivation}.

The posterior update procedure is implemented in the function
${\textsc {updateDistribution}}(F_{i,t-1},y_{it})$ shown in
Fig.~\ref{fig:updateDist}. Its inputs are the posterior distribution $F_{i,t-1}$ after
round $t-1$ and the bit $y_{it}$ sent by user $i$ to the CA at the beginning of round $t$.
Its output is the updated posterior $F_{it}$.
\begin{figure}
\fbox{
 \addtolength{\linewidth}{-2\fboxsep}%
 \addtolength{\linewidth}{-2\fboxrule}%
\begin{minipage}{\linewidth}
\begin{spacing}{1}
\begin{algorithmic}
{
\Function{$F_{it}\gets$\textsc{updateDistribution}}{$F_{i,t-1},y_{it}$}
\If {$(y_{it}==1)$}
 \State Compute $F_{it}$ using (\ref{eqn:update11}) and (\ref{eqn:update12}).
\Else
 \State Compute $F_{it}$ using (\ref{eqn:update01}) and (\ref{eqn:update02}).
\EndIf   
\EndFunction
}
\end{algorithmic}
\end{spacing}
\end{minipage}
}
\caption{Function to update the posterior distribution.}
\label{fig:updateDist}
\end{figure}
\subsection{Convergence and asymptotic optimality}
\label{sec:matched_convergence}
\newtheorem{lemma}{Lemma}
\begin{lemma}[]
Under matched auctioning, the posterior median of each bid price converges to the respective bid price in probability.
\label{lemma:lemma1}
\end{lemma}
\begin{proof}
This result is based on the proof of Theorem 1 on page 3 of~\cite{ShFe08},
where it is shown that under posterior matching, the posterior distribution of the message
point computed by the receiver (in our case, the CA) converges in probability to the unit step at
the actual message point sent by the transmitter (in our case, the $i^{\text{th}}$ user).
Using the notation from (\ref{eqn:pdist}) for the CA's posterior distribution of
the $i^{\text{th}}$ user's bid, we have the following for any $\omega>0$ and $\delta>0$:
$${\bf P}\{|F_{it}(B_i+\delta)-F_{it}(B_i-\delta)-1| < \omega\} \rightarrow 1,$$
where the probability is evaluated using the joint distribution of CA's prior model for $B_i$
and the channel outputs.  Using $\omega=1/2$, we have that for any $\delta > 0$ there
exists $t_0>0$ such that for all $t>t_0$,
\[
{\bf P}\{ | F_{it}(B_i + \delta) - F_{it}(B_i - \delta)  - 1 | < 1/2 \} > 0.
\]
Equivalently,
\begin{equation*}
{\bf P}\{ 1/2 < F_{it}(B_i + \delta) - F_{it}(B_i - \delta)  < 3/2 \} > 0.
\end{equation*}
This implies that for any $\delta > 0$ there exists $t_0>0$ such that for all $t>t_0$,
\begin{equation}
{\bf P}\{ F_{it}(B_i + \delta) > 1/2 + F_{it}(B_i - \delta)  \} > 0.
\label{eq:lemma1}
\end{equation}
Since $F_{it}$ is a cumulative probability distribution, both $F_{it}(B_i + \delta)$ and
$F_{it}(B_i - \delta)$ must be between 0 and 1 with probability 1.  Therefore, it follows
from (\ref{eq:lemma1}) that, for all $t>t_0$, $F_{it}(B_i - \delta) < 1/2$
and $F_{it}(B_i + \delta) > 1/2$ with probability 1.  To see this, suppose that there is
a non-zero probability that $F_{it}(B_i - \delta) \geq 1/2$ for some value $t > t_0$.
Then (\ref{eq:lemma1}) would imply a non-zero probability for $F_{it}(B_i + \delta) > 1$,
which is a contradiction.  A similar argument shows that $F_{it}(B_i + \delta) > 1/2$ with
probability 1.

But the fact that $M_{i,t+1}$ is
the median of $F_{it}$ means that $F_{it}(M_{i,t+1}) = 1/2$.  So, we get:
\[
F_{it}(B_i - \delta) < F_{it}(M_{i,t+1}) < F_{it}(B_i + \delta),
\]
for all $t>t_0$, with probability 1. This implies
\[
B_i - \delta < M_{i,t+1} < B_i + \delta,
\]
with probability 1 and
\[
{\bf P}( | M_{i,t+1} - B_i | < \delta) = 1,
\]
for all $t>t_0$.  Since such $t_0$ exists for any $\delta$, this shows that $M_{it}$ converges
to $B_i$ in probability as $t\rightarrow\infty$, completing our proof of Lemma~\ref{lemma:lemma1}.
\end{proof}

\begin{proposition}[]
For any $h>0$, the probability of allocating spectrum to the highest bidder converges to one and the CA's revenue converges to $B_{(N)}-h$ in probability, where $B_{(N)}$ is the largest
of the $N$ bids.
\label{lemma:lemma2}
\end{proposition}
\begin{proof}
We showed in Lemma~\ref{lemma:lemma1} that the posterior median $M_{ij}$ will eventually be
within $\delta$ of the bid price $B_i$ with probability 1, for any $\delta$.
Suppose that, for two users $i$ and $j$, we have the following realizations of
the bid prices: $B_i = b_i$ and $B_j = b_j$, and $b_i > b_j$.  Then, applying
Lemma~\ref{lemma:lemma1} with $\delta = (b_i - b_j)/2$, we see that there exists $t_{ij}$
such that for any $t>t_{ij}$, the posterior median $M_{it}$ is larger than
the posterior median $M_{jt}$ with probability 1.  Now for a particular realization
of the bid prices, let $k$ be the index of the maximum bid price, and let
\[
t_1 = {\max}_{j:j\neq k}\{ t_{kj}\}.
\]
(Without loss of generality, we are assuming here that only one bidder has
the maximum bid.) Then for any $t > t_1$, the posterior median $M_{kt}$
corresponding to the maximum bid will be larger than any other posterior
median, with probability 1.  Therefore, for any $t > t_1$, spectrum will be
awarded to the highest bidder with probability 1, and the CA's revenue will
only depend on the posterior median of the highest bidder.\footnote{It must be noted
here that $t_1$ will depend on the particular realization of the bid prices:
if two largest bids are very close, then it would take a large $t_1$ for their
respective posterior medians to get ordered correctly with
probability 1. On the other hand, if two largest bids are very close,
then awarding spectrum to the second-highest bidder would be nearly optimal
for the CA.}

We denote the highest bid and the corresponding median after the $t^{\text{th}}$ update
using $b_{(N)}$ and $M_{(N),t+1}$ respectively. Recall that we set the ask price to
${\max}_i  \{M_{i,t+1}\} - h$, where $h>0$. From Lemma~\ref{lemma:lemma1} and
the preceding paragraph, for any $\delta>0$, we have that there exists $t_2 > t_1$ such that for all $t > t_2$,
$${\bf P}\{b_{(N)}-h-\delta < M_{(N),t+1}-h < b_{(N)}-h+\delta\} = 1.$$
\\
{\em Case (i):} If $\delta \leq h$, for all $t > t_2$,
$${\bf P}\{ b_{(N)}-h-\delta < \text{Ask price at time $t$} < b_{(N)}-h+\delta \} = 1$$
Since the ask price is smaller than the maximum bid, this is equivalent to
$${\bf P}\{ b_{(N)}-h-\delta < \text{Revenue at time $t$} < b_{(N)}-h+\delta \} = 1.$$
For $\delta = h$, let the corresponding time $t_2$ be $t_h$.
\\
{\em Case (ii):} If $\delta > h$,
$${\bf P}\{ b_{(N)}-h-\delta < \text{Revenue at time $t$} < b_{(N)}-h+\delta\} \geq $$
$${\bf P}\{ b_{(N)}-2h < \text{Revenue at time $t$} < b_{(N)}\}.$$
But if $t>t_h$, then 
$${\bf P}\{ b_{(N)}-2h < \text{Revenue at time $t$} < b_{(N)}\} = 1.$$
Therefore, for this case, we can pick any $t_2 > t_h$.
Combining these two cases, we have that revenue at time $t$ converges to $b_{(N)}-h$
in probability, where all the probabilities are computed conditioned on a realization
of bids. Since this is true for any realization of the bid prices, we can remove
the conditioning on the bids to obtain that revenue at time $t$ converges to
$B_{(N)}-h$ in probability, where this probability is over the joint distribution
of bids and channel realizations. The last operation, where we exchange the integral
with respect to the joint density of the bids and the limit on $t$, is possible due
to the dominated convergence theorem~\cite{Ro88}. \footnote{The revenue converges to $b_{(N)}-h$
for all bid price realizations, probability is bounded and the joint density of
the bids is also bounded and well defined. Therefore, the conditions needed to apply
dominated convergence hold.} Although $h$ is an arbitrarily small
positive number, a smaller value of $h$ would result in a larger
number of rounds for the CA's revenue to converge to $B_{(N)}-h$.
In other words, the more the CA is willing to give up in revenue compared 
to the highest bid, the more quickly it's revenue would converge to $B_{(N)}-h$.
\end{proof}

\section{Single-unit auctions with strategic users: Truthful matched auction}
\label{sec:truthful_matched_auction}
In our exposition so far, we have assumed that users are not strategic. 
Strategic users act rationally and aim to maximize their payoff.
So using matched auctioning for strategic users may lead to inefficient
allocations, where the user who values spectrum the most is not allocated 
the resource even after many update-and-allocate rounds. In this section, 
we address the issues posed by strategic bidders
by extending the current matched auction set-up to a {\em truthful matched} auction. 
We know from~\cite{Kr02} that for a standard
auction where the winner pays the second highest bid, bidding truthfully is a
weakly dominant strategy for strategic users who want to maximize their payoff. Truthful matched auctioning tries to replicate a second price sealed bid auction under the current set-up.
To recollect, in the matched auction setup, the CA maintains posterior 
distributions of the bids of each user. In each round, the CA awards spectrum to the user with the highest posterior median and sets the ask price close to the highest posterior median, which the winner can compute.

In truthful matched auctioning, identifying the highest bidder works in the same way as in matched auctioning. Additionally, the CA and all the users maintain an ask price distribution ($A_t$), whose median ($a_t$) is taken to be the ask price for that round.
The CA has to send the bits it receives from the second highest bidder back to the users, 
so that they can compute the ask price from the posterior distribution of the second 
highest bid price. But at the outset, the CA does not know who the two highest bidders are.
To overcome this difficulty, the CA treats the second highest posterior median as the message point at each step.
It sends one additional bit to each user, denoted using $\tilde{y}_{it}$, which is equal 
to 1 if the second highest posterior median is larger than the median of the ask price distribution and 0 otherwise. To recall, this is the same
$\tilde{y}_{it}$ introduced in Section~\ref{sec:scheme_outline} and
illustrated in Fig.~\ref{fig:auctionOutline}.
\begin{equation}
\tilde{y}_{it} = I_{\text{Second highest posterior median $>$ Median of ask distribution}}
\label{eq:tildey}
\end{equation}
The modification is shown in Fig.~\ref{fig:truthful_flow}.
\begin{figure}[htb]
\begin{minipage}[b]{1.0\linewidth}
  \centering
  \centerline{\epsfig{figure=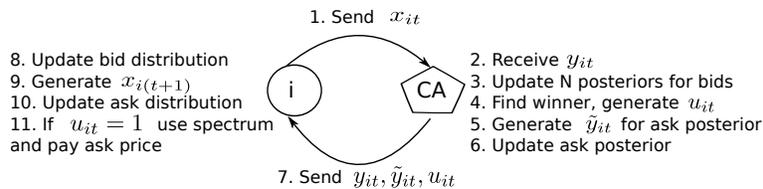,width=10cm}}
\end{minipage}
\caption{Actions taken by user $i$ and CA in round $t$.}
\label{fig:truthful_flow}
\end{figure} 

Posterior updates on the ask distribution are carried out by the CA and by all the users as
if this bit has been received from a {\em{virtual user}} through a BSC$_p$ in each round. 
So, the virtual user sends the bit $\tilde{y}_{it}$ according to the position of a message point--which is equal to the second highest posterior median at that round--and a posterior distribution--which is equal to the distribution of the ask price. The update equations are the same as given in Proposition~\ref{prop:prop1}, where we replace $y_{it}, m_{it},
m_{it}', F_{i,t-1}, F_{it}$ with $\tilde{y}_{it}, a_{it}, a_{it}', A_{i,t-1}, A_{i,t}$ respectively. As in the case of the bid distributions, the ask price distribution is assumed to have a uniform prior 
over $[0,1]$. In Section~\ref{sec:truth_results}, we show by simulation that in
truthful matched auctioning, the revenue of the CA tends to the second highest bid price in probability. 

We have previously mentioned in Section~\ref{sec:sysset} that 
strategic winners also pose the problem of using spectrum and 
claiming to have rejected it.
This could happen when the users are base stations that do not need 
the help of the CA in order to communicate during the spectrum use period. In this section, we avoid this problem by not allowing winners 
to reject spectrum. This is to assume that winners are always able to 
pay the ask price, even if it is larger than their 
private values. Although winners may initially have to pay a price 
more than their private values, we shall see using simulations that as the number of update-and-allocate rounds increases, the winner's payoff converges to the true theoretical payoff in probability. Since this replicates
a second price sealed bid auction, the payoff of the highest bidder
converges in probability to the absolute difference between the two highest bids. The payoff of the other bidders converges in probability to zero.

\section{Multi-unit auctions with strategic users: Quantized Vickrey auctions}
\label{sec:vickrey_auction}
In current wireless standards, there are hard constraints on the number of 
bits per CQI
(channel quality information). In general, if there 
are $K$ bits per CQI, then we can extend the preceding one-bit truthful
matched auction scheme into 
a $K$-bit scheme to simultaneously auction $K$ sub-channels. The usual assumption 
in multi-unit auctions is that the units are all identical, and they have 
diminishing marginal values for every user. A user who wants to bid for $K$ 
sub-carriers would, instead of having a single value, have a value profile given by 
${\mathbf {V}_i} = (v_{i1}, \cdots, v_{iK})$ such that $v_{ik} \geq v_{i(k+1)}$. 
The components of the vector ${\mathbf{V}_i}$ specify marginal values, 
which means user $i$'s value for one unit of spectrum is $v_{i1}$. For two units, it 
is $v_{i1}+v_{i2}$ and so on. We assume that the $i^{\text{th}}$ 
value profile is jointly distributed as the order statistics of $K$ 
random variables that are i.i.d. uniform over $[0, 1]$.

A strategy in this case is $(s_{i1}(v_{i1}),\cdots,s_{iK}(v_{iK})) = ({b_{i1},\cdots,b_{iK}})$, 
which maps the value vector into a bid vector. The components of the bid vector
are equal to marginal bid prices. The multi-unit analogue 
of a second price auction is called Vickrey auction,
where the top $K$ bids are awarded one unit of spectrum each, and if
user $i$ is awarded $k_i$ units of spectrum, then it is charged
an amount equal to the $k_i$ highest losing bids excluding its
own bids. The payoff of the user $i$ is therefore its value for $k_i$
units minus the ask price that it is charged.
For example, if there are 4 units of spectrum and 4 
users with bid profiles 
(21, 15, 5, 3), (32, 18, 15, 10), (25, 23, 15, 12) and (30, 20, 10, 8),
then the top 4 bids are 32, 30, 25, 23. So user 1 gets zero units,
user 2 gets one unit, user 3 gets two units and user 4 gets one unit.
In this example, user 1 gets nothing and pays nothing. 
User 2 pays 21,
user 3 pays 21+20 and user 4 pays 21 as per the payment rules stated before.
The payoffs of the users are respectively equal to 
0, 11, 7 and 9.
It can be shown that for a Vickrey auction, the truthful strategy given by 
$(s_{i1}(v_{i1}),\cdots,s_{iK}(v_{iK})) = ({v_{i1},\cdots, v_{iK}})$ is weakly dominant~\cite{Kr02}.

Quantized Vickrey auctions can be implemented along the same lines as 
truthful matched auctions, but with a few more modifications. The first
difference here is that for each user, posterior updates are carried out on
each of the $K$ marginal bids, which are distributed apriori as the order
statistics of $K$ independent and uniform random variables over $[0,1]$.
Secondly, each user has a separate ask price that can take values in $[0,K]$. 
So each user maintains and updates a separate ask price distribution that is apriori uniform over $[0,K]$. Thirdly, in each round, the $i^{\text{th}}$ 
user sends the $K$-bit vector ${\bf{x}}_{it}$, whose components are calculated using the corresponding marginal bid and the corresponding posterior median. More explicitly, 
the $k^{\text{th}}$ component of ${\bf{x}}_{it}$ is equal to 
$$I_{\text{$k^{\text{th}}$ marginal bid of user $i$ }>\text{ Posterior median of $k^{\text{th}}$ marginal bid of user $i$}}.$$
These $K$ bits are used to update the $K$ marginal bid posteriors of user $i$.
Fourthly, the CA sends $2K+1$ feedback bits to each user. The first $K$ of these 
bits are equal to ${\bf{y}}_{it}$, whose components are the $K$ received bits
so that the users can perform the same updates as the CA.
The next $K$ bits are denoted using ${\bf{u}}_{it}$. The $k^{\text{th}}$ component 
of ${\bf{u}}_{it}$ denotes whether the $i^{\text{th}}$ user won the $k^{\text{th}}$ unit of spectrum or not. The updated bid price estimates are used to decide the winners and the ask price for that spectrum use period, 
as per the rules of Vickrey auction outlined earlier. The last bit of feedback ($\tilde{y}_{it}$) is used to convey the updated ask distribution to the user.
\begin{equation}
\tilde{y}_{it} = I_{\text{User $i$'s ask price in round $t >$ Median of user $i$'s ask distribution in round $t$}}
\label{eq:tildey2}
\end{equation}
While updating the ask price posterior, this bit is again treated by the user and by
the CA as if it has been received from a virtual user through a BSC$_p$. Apart from these four changes, the update procedure and update equations are identical to the truthful 
single-unit auctions.
In Section~\ref{sec:truth_results}, we show by simulation that in our
implementation of quantized Vickrey auctions, the revenue of the CA 
and the payoffs of the users tend to the true theoretical revenue 
and the true payoff in probability. This in turn shows that asymptotically, 
it is weakly dominant for each user to truthfully reveal its value profile.
\section{Matched auctioning with time-varying bids}
\label{sec:movingBidsScheme}
In this section, we return to the matched auctioning algorithm with non-strategic users,
and extend it for time-varying bids.
When the bids are allowed to vary with time, it is possible under matched auctioning, 
that the CA becomes overconfident about its estimates of the bid prices. By this we 
mean that when a user's bid changes after remaining constant for many update-and-allocate rounds, 
the posterior distribution of the bid would be very close to the unit step at the 
previous price. Consequently, the corrections sent by the user would not affect the 
CA's estimate of the price significantly. This would result in the matched auctioning 
algorithm tracking bid prices very slowly, or completely failing to track them. 
So, the CA's revenue could be substantially
lower than the highest bid. In this 
section, the system set-up and the auction scheme are the same as in matched
auctioning, except for the bid-drift model and a significant modification to the posterior 
update step in Section~\ref{sec:matched_posteriorUpdates}. 
\subsection{Bid-drift model}
\label{sec:bidDriftModel}
For each user $i$, bids are represented as independent discrete-time 
random processes $B_i(t)$, and an additive model is used to represent their dynamics.
{
\begin{equation}
 \ B_{i}(t+1)=
\begin{cases}
\min\{\max\{B_i(t)+n_i(t+1), 0 \},1\} & \text{w. p. } q
\\
B_i(t) & \text{w. p. } 1-q
\end{cases}
\label{eqn:drift}
\end{equation} 
}

\noindent
In (\ref{eqn:drift}), $B_i(1)$ is uniformly distributed in $[0,1]$, \mbox{$n_i(t+1)$} is independently uniform over 
a small interval $[-\epsilon,\epsilon]$ and $q$ is the probability that the bid price changes
in the next round. So for any update round $t$ and user $i$,
\mbox{${\bf P}\{ i^{\text{th}}\text{ bid changes at } t+1 \} = q$}, and $q$ is assumed to
be constant and identical for all the users.

\subsection{Posterior update algorithm for time-varying bids}
\label{sec:modifiedPM} 

After each posterior distribution update,
if the posterior distribution of any bid price is sufficiently close to the unit step function
at the respective posterior median, then the CA approximates the distribution as another
distribution that is more spread-out than the unit step function.  Therefore, our main
idea to enable the CA to track moving bids is to perform posterior updates while not
allowing the individual posterior distributions to collapse into the unit step function.
The distribution that we use for approximation must be such that most of the corresponding
density is concentrated about the posterior median but at the same time,
all values in $[0,1]$ have non-zero density. Although the approximation comes
at the cost of the CA not knowing exactly what the bid price is, we show by simulation
that this is effective in achieving revenues that are close to the maximum bid
and outperforms matched auctioning when the bids are time-varying.

As an approximation of the unit step at $b_0 \in [0,1]$, we take a cumulative distribution
function $F(b;b_0,\lambda,\mu)$ with median $b_0$. The corresponding probability density
is denoted using 
$f(b;b_0,\lambda,\mu)$. The shape parameters of the distribution,
$0< \lambda \ll 1$ and $0<\mu \ll 1$, control how close $F$ is to the unit step at $b_0$.
We use a piecewise linear $F$, whose shape is illustrated in Fig.~\ref{fig:Fb0} for a few parameter values. The exact equations for the three cases depicted in
Fig.~\ref{fig:Fb0} are shown in (\ref{eqn:pwl1}), (\ref{eqn:pwl2}) and (\ref{eqn:pwl3}).
\begin{figure}[htb]
\begin{minipage}[b]{1.0\linewidth}
  \centering
  \centerline{\epsfig{figure=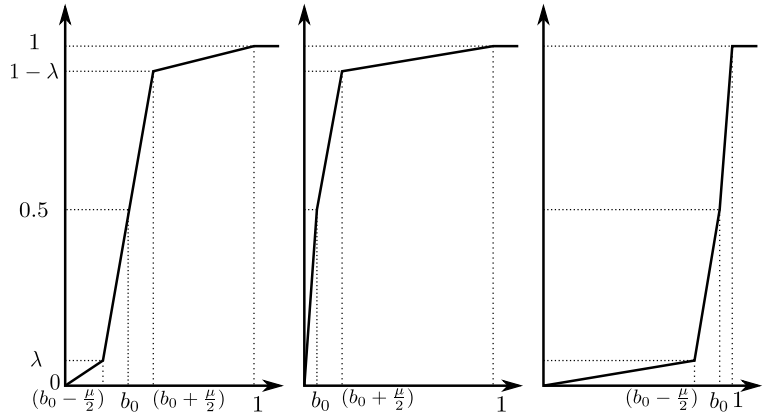,width=10cm}}
\vspace*{-0.25cm}
\end{minipage}
\caption{Shape of $F(b;b_0,\lambda,\mu)$ for 3 possible locations of $b_0$.}
\label{fig:Fb0}
\end{figure}

\noindent
Case 1: If $b_0-\frac{\mu}{2} >0$ and $b_0+\frac{\mu}{2} < 1$, then $F(b;b_0,\mu,\lambda)=$
\begin{equation}
 \
\begin{cases}
\frac{\lambda}{(b_0-\frac{\mu}{2})}b & \text{ if $b \in (0,b_0-\frac{\mu}{2}]$}
\\
 (b-b_0+\frac{\mu}{2})\frac{1-2\lambda}{\mu}+ \lambda & \text{if $b \in (b_0-\frac{\mu}{2}, b_0+\frac{\mu}{2}]$ }
\\
\frac{\lambda}{1-b_0-\frac{\mu}{2}}(b-b_0-\frac{\mu}{2})+1-\lambda & \text{if $b \in (b_0+\frac{\mu}{2},1)$.}
\end{cases}
\label{eqn:pwl1}
\end{equation}
Case 2: If $b_0-\frac{\mu}{2} <0$, then $F(b;b_0,\mu,\lambda)=$
\begin{equation}
 \
\begin{cases}
\frac{b}{2b_0} & \text{ if $b \in (0,b_0]$}
\\
 (b-b_0)\frac{1-2\lambda}{\mu}+ \frac{1}{2} & \text{if $b \in (b_0, b_0+\frac{\mu}{2}]$ }
\\
\frac{\lambda}{1-b_0-\frac{\mu}{2}}(b-b_0-\frac{\mu}{2})+1-\lambda & \text{if $b \in (b_0+\frac{\mu}{2},1)$.}
\end{cases}
\label{eqn:pwl2}
\end{equation}
Case 3: If $b_0+\frac{\mu}{2} >1$, then $F(b;b_0,\mu,\lambda)=$
\begin{equation}
 \
\begin{cases}
\frac{\lambda}{b_0-\frac{\mu}{2}}b & \text{ if $b \in (0,b_0-\frac{\mu}{2}]$}
\\
 (b-b_0+\frac{\mu}{2})\frac{1-2\lambda}{\mu}+ \lambda & \text{if $b \in (b_0-\frac{\mu}{2}, b_0]$ }
\\
\frac{b-b_0}{2(1-b_0)}+\frac{1}{2} & \text{if $b \in (b_0,1)$.}
\end{cases}
\label{eqn:pwl3}
\end{equation}

We assume that at the end of round $t$, the posterior density of $B_i$ is $f_{it}$,
with median $m_i$. For a threshold $\theta>0$, if $D(f_{it}(b) || f(b;m_i,\lambda,\mu)) < \theta$, then we approximate the posterior distribution using $F(b;m_i,\lambda,\mu)$. As a measure of divergence ($D$) we use the Bhattacharyya distance~\cite{Ka67}. 
If $f_1(b)$ and $f_2(b)$ are probability density functions
of continuous random variables, then the Bhattacharyya distance between them is given by
\begin{equation}
D(f_1(b)||f_2(b))=-\log\left(\int \sqrt{f_1(b)f_2(b)} db\right)
\label{eqn:bhattc}
\end{equation}
In the case of discrete probability mass functions $p_1$ and $p_2$, the Bhattacharyya distance between the two mass
functions is
\begin{equation}
D(p_1(b)||p_2(b))=-\log\left(\sum\limits_b\sqrt{p_1(b)p_2(b)}\right)
\label{eqn:bhattd}
\end{equation}

If we use \textsc{updateTrack} for posterior updates, it is possible for the posterior 
median to be substantially larger than the bid price even after many update-and-allocate rounds have been completed. 
This means that if we set the ask price very close to the posterior median of the winner, 
then the winner could  reject spectrum even after many rounds. Therefore, the CA sets the 
ask price to be equal to ${\max}_i  \{M_{i,t+1}\} - \mu$, where $\mu$ is the 
same as shown in Fig.~\ref{fig:Fb0}. In the following round, the $i^{\text{th}}$ user 
replies back to the CA, knowing that the CA's estimate of the posterior median and consequently 
the ask price, is based on this modified procedure. Therefore, the ability to 
track moving bid prices comes at the cost of setting the ask price to a value that is lower 
than in the case of constant bid prices. But in the next section, we see that despite the 
approximation and setting the ask price low, the revenue of the CA is still very close to 
the ideal case of perfect tracking. Therefore, to adjust for tracking, we simply replace
{\sc updateDistribution} in Figs.~\ref{fig:overallAlgo} and~\ref{fig:oneround} with {\sc updateTrack} shown in Fig.~\ref{fig:pcodeDrift} and use $\mu$ in place of $h$ while setting the ask price.
\begin{figure}
\fbox{
 \addtolength{\linewidth}{-2\fboxsep}%
 \addtolength{\linewidth}{-2\fboxrule}%
\begin{minipage}{\linewidth}
\begin{spacing}{1}
\begin{algorithmic}
{
\Function{$F_{it}\gets$\textsc{updateTrack}}{$F_{i,t-1},y_{it}$}

  \State $F_{it} \leftarrow \textsc{updateDistribution}(F_{i,t-1},y_{it})$
  \State $m_{i,t+1} \gets \text{Median}(F_{it})$
  \State $f_{it} \gets$ density corresponding to $F_{it}$
  \If{$D\big(f_{it}(\cdot)||f(\cdot| m_{i,t+1},\lambda,\mu)\big)<\theta$}  
  
    $F_{it}(\cdot)\gets F(\cdot\ ; m_{i,t+1},\lambda,\mu) $
  \EndIf
\EndFunction
}
\end{algorithmic}
\end{spacing}
\end{minipage}
}
\caption{Posterior updates adjusted for bid-drift.}
\label{fig:pcodeDrift}
\end{figure}

\section{Simulation set-up and results}
\label{sec:simulations}
In Section~\ref{sec:matched_posteriorUpdates}, we have derived recurrence relations for updating
the posterior distribution of bid prices. But deriving these updates in closed form
is difficult even for simple prior distributions. In order to circumvent this problem
and to account for prices being discrete, we discretize the $[0,1]$ interval such that bid prices are
integer multiples of $\Delta$.
\subsection{Convergence of matched auctioning}
\label{sec:simconv}
Fig.~\ref{fig:matchedConv10} shows the convergence of the revenue to $B_{(N)}-h$, with $\Delta=10^{-5}$ and $h=10^{-3}$. Here again, $B_{(N)}$ is used to
denote the maximum of the $N$ bids.
We show for a few values of $\delta$, that \mbox{${\bf P}\{|\text{Revenue at time }t - B_{(N)}+h| < \delta \} \rightarrow 1$}, where probabilities are estimated as empirical probabilities over $R=1000$ independent joint realizations of bid prices and channel outputs. As opposed to this behavior, we see that the mean revenue for our 
first scheme of unmatched auctioning explained in 
Section~\ref{sec:unmatched} is significantly smaller than 
the mean maximum bid. This is shown in Fig.~\ref{fig:unmatched} along 
with standard error bars.
The length of the error bars in each round is equal to ${2\hat{\sigma}}/{\sqrt{R}}$, where $\hat{\sigma}$ is the estimated standard deviation of the revenue in that particular round, and $R=10^4$ is the number of Monte-Carlo rounds over which the averaging was performed.
\begin{figure}[htb]
\begin{minipage}[b]{1.0\linewidth}
  \centering
  \centerline{\epsfig{figure=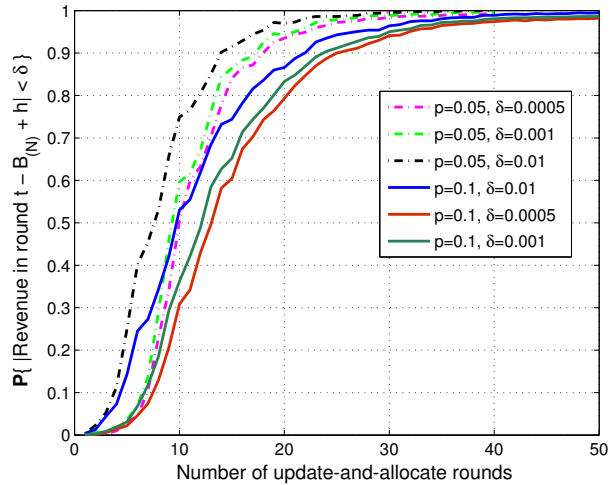,width=9cm}}
\vspace*{-0.25cm}
\end{minipage}
\caption{Convergence of revenue for matched auctioning with $N=10$.}
\label{fig:matchedConv10}
\end{figure}

\begin{figure}[htb]
\begin{minipage}[b]{1.0\linewidth}
  \centering
  \centerline{\epsfig{figure=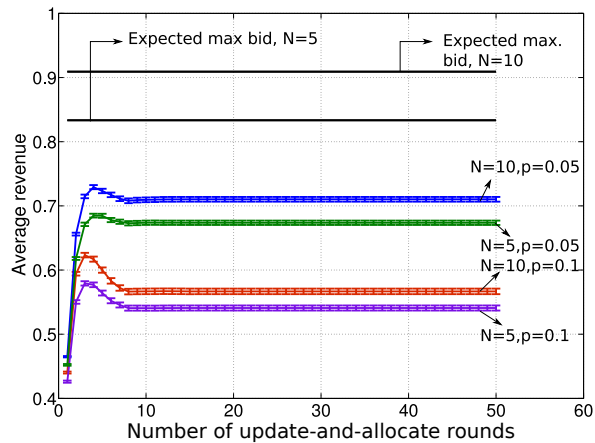,width=9cm}}
\vspace*{-0.25cm}
\end{minipage}
\caption{Average revenue of unmatched auctioning with error bars for different values of $N$ and $p$.}
\label{fig:unmatched}
\end{figure}
\subsection{Convergence of auctions with strategic users}
\label{sec:truth_results}
From Sections~\ref{sec:truthful_matched_auction} and~\ref{sec:vickrey_auction},
we can see that truthful matched auctioning is
exactly the same as quantized Vickrey auctions for $K=1$.
We show by simulation the convergence properties of quantized Vickrey auctions for 10 users and $K=1$ or $K=4$ units.
The left panel of Fig.~\ref{fig:truthful_revenue_payoff} illustrates 
the convergence of the CA's revenue to the true theoretical revenue 
in probability. Unlike matched auctioning
where all the probabilities start very close to zero, 
we sometimes have non-zero probabilities starting at the very 
first round since in Vickrey auctions, we do not allow winners 
to reject spectrum.
We also show that the average user payoff converges to the true theoretical average, where the averaging is done
over users. This is depicted in the right
panel of Fig.~\ref{fig:truthful_revenue_payoff}. From this figure
we infer that asymptotically, quantized Vickrey auctions behave
identically to Vickrey auctions. Therefore, as the number of
auction rounds increases, it is weakly dominant for the users
to reveal their bids truthfully.
\begin{figure*}[htb]
\begin{minipage}[b]{1.0\linewidth}
  \centering
  \centerline{\epsfig{figure=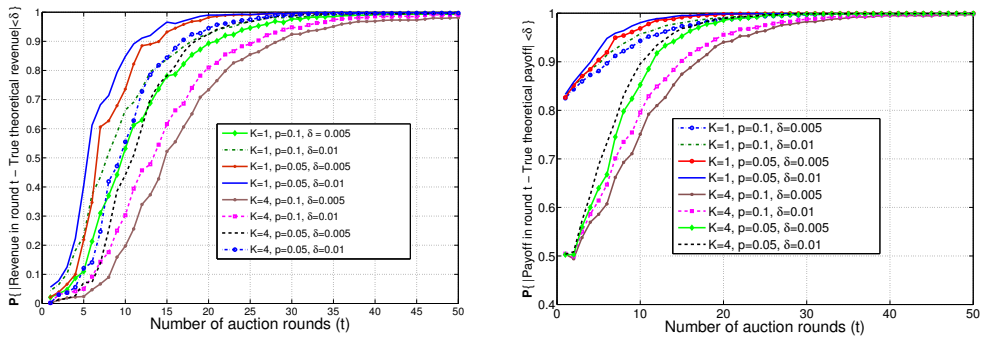,width=14cm}}
\vspace*{-0.25cm}
\end{minipage}
\caption{Convergence of revenue and payoff to their true theoretical 
values in quantized Vickrey auctions for $N=10$ users.}
\label{fig:truthful_revenue_payoff}
\end{figure*}
\vspace*{-0.5cm}
\subsection{Comparison between matched auction and matched auction adjusted for bid-drift}
\label{sec:trackVnot}
When bids are allowed to vary with time according to the model proposed in Section~\ref{sec:bidDriftModel}, we show that the tracking algorithm proposed in Section~\ref{sec:modifiedPM} performs better than
the case when there is no tracking.
In order to compare the tracking algorithm with matched auctioning, 
we compare the {\em efficiency ratio}, which is the ratio of the average 
revenue to the average 
maximum bid, averaged over the update-and-allocate rounds and then averaged over Monte Carlo rounds. 
If the CA is able to perfectly track the bids, then we expect this ratio to be 1 for 
all values of $q$ and $\epsilon$, defined in Section~\ref{sec:bidDriftModel}.
But in reality, for $N=5$, we observe the behavior shown in Fig.~\ref{fig:trackVsNot}. The left and right
panels of this figure depict the ratio as a function 
of $q$ for $\epsilon=0.01$ and as a function of $\epsilon$ for $q=0.02$, respectively.
In these experiments we use $\Delta=1/5000$, $h=1/1000$, and for the piecewise linear approximation shown in Fig.~\ref{fig:Fb0}, we take the parameters to be $\lambda=0.005$ and $\mu=0.005$. We take the threshold on the Bhattacharyya distance $\theta=0.3$. 
\begin{figure*}[htb]
\begin{minipage}[b]{1.0\linewidth}
  \centering
  \centerline{\epsfig{figure=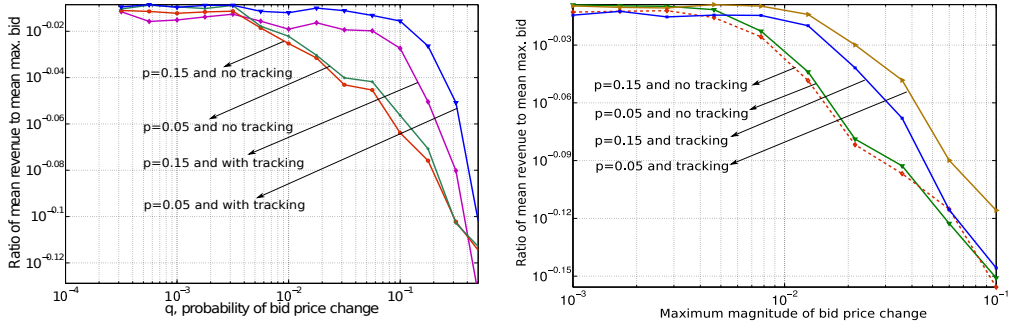,width=14cm}}
\vspace*{-0.25cm}
\end{minipage}
\caption{${\text{(Avg. revenue)}}/{\text{(Avg. maximum bid)}}$ vs. $q$, for $\epsilon=0.01$ (left panel), and vs. $\epsilon$,
for $q=0.02$ (right panel)}
\label{fig:trackVsNot}
\end{figure*}
We see from the left panel of Fig.~\ref{fig:trackVsNot} that 
with tracking, the ratio is still very close to 1 till around $q=0.1$.
We also observe that the improvement achieved by matched auctioning with tracking is most pronounced when $q$ is between $10^{-2}$ and $10^{-1}$. For very small values of $q$, the bids do not drift too much and the two methods are equally good. In contrast, for very large values of $q$, both methods are unable to track effectively.
The right panel of Fig.~\ref{fig:trackVsNot} also shows 
matched auctioning with tracking performing
better than without tracking for \mbox{$\epsilon > 5\times10^{-3}$}.
\subsection{Sensitivity of tracking algorithm to parameter settings}
\label{sec:paramsweep}
In this subsection, we examine the sensitivity of the efficiency ratio with respect to
the tracking parameters $\lambda$, $\theta$ and $\mu$. For these simulations, we fix $N=5$, $p=0.05$ and
sweep over one of the parameters while keeping the other two constant. For the case where we sweep 
$\lambda$ over the interval $[0.001, 0.01]$, the results show very little sensitivity to the value of $\lambda$ as seen in 
the left panel of Fig.~\ref{fig:sweepLambdaThetaMu}.
\begin{figure*}[htb]
\begin{minipage}[b]{1.0\linewidth}
  \centering
  \centerline{\epsfig{figure=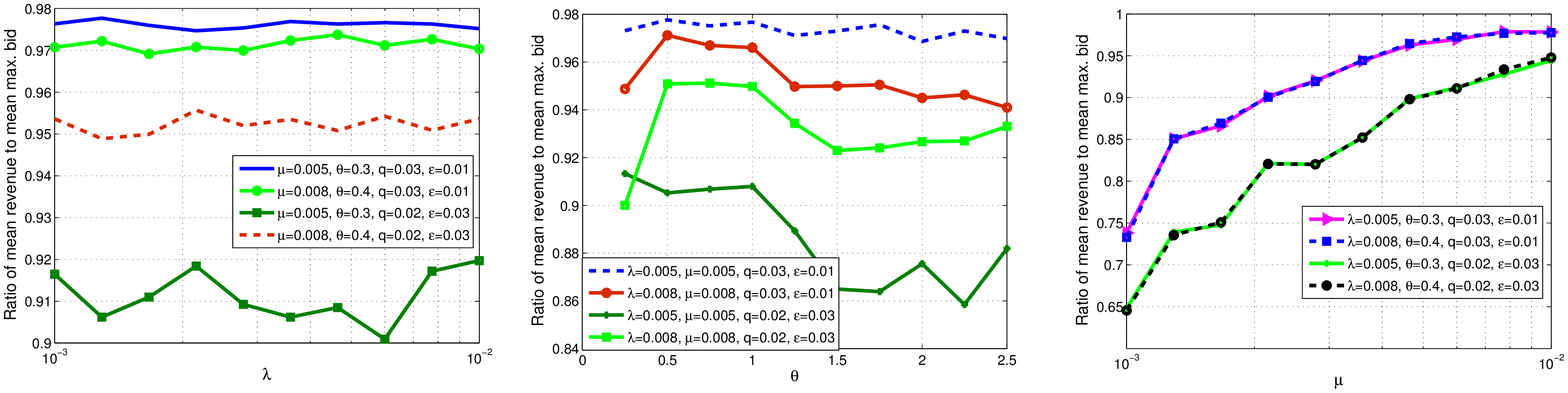,width=17cm}}
\vspace*{-0.25cm}
\end{minipage}
\caption{${\text{(Avg. revenue)}}/{\text{(Avg. maximum bid)}}$ vs. $\lambda$ (left panel), $\theta$ (middle panel) and $\mu$ (right panel) for $N=5$ and $p=0.05$.}
\label{fig:sweepLambdaThetaMu}
\end{figure*}
Similarly, the plot in the middle panel of Fig.~\ref{fig:sweepLambdaThetaMu}
does not show much sensitivity to $\theta$.
In contrast to these results, while sweeping $\mu$ over $[0.001, 0.01]$, we observe that for small
values of $\mu$, the performance is significantly degraded as seen in the right panel of Fig.~\ref{fig:sweepLambdaThetaMu}. 
The reason for this is that low values of $\mu$ result in $F$ being too close to the unit step, which causes tracking to be very slow. Moreover, small values of $\mu$ would result 
in more rejections by the winner if the bid price happens to decrease. 
\vspace*{-0.25cm}

\subsection{Comparison of overheads in open-loop and closed-loop schemes}
\label{sec:appendix_openVsclosed}
The bits used for bid revelation go through a noisy channel. Hence
they must be encoded to correct errors.
Without encoded bits, the scheme would correspond to
unmatched auctioning which is highly inefficient as seen in Section~\ref{sec:simconv}.
The encoding can either use channel output feedback, like we have done
in our scheme, or ignore feedback and use an open-loop scheme.
We argue here that if a state-of-the-art open-loop scheme is used,
then the number of bits required for 
reliable bid revelation is much higher than for our proposed scheme.
For simplicity, we consider the case of a single unit auction 
and a single user revealing its bid to the CA.
Open-loop schemes in the form of block codes encode a message
that is $k$ bits long into a codeword of $n>k$ bits. The receiver
uses a noisy received codeword to determine which one of the $2^k$ 
equally likely messages was sent.
The message point perspective cannot be directly compared with the 
block coding perspective since the message point is a continuous random
variable, whereas the message in block coding is discrete. Nevertheless, 
we endeavor to develop a fair 
comparison by mapping the message point perspective to the block coding perspective here.
For simulating posterior matching, we divide the $[0,1]$
interval into $K$ equally spaced intervals. Now the message point
is a discrete random variable whose observed value can be one of
$K$ equally likely messages. If we use
$n$ transmissions, then we define the rate of the code to be $\log_2(K)/n$.
Popular open-loop codes like LDPC codes use the definition of rate to be $k/n$,
where the message is $k$ bits long and the codeword is $n$ bits long. The
code is evaluated using the frame error rate (FER), which is the probability
that a message would get decoded into an incorrect message by the receiver. 
Similarly, in the message point perspective, we say that a frame error 
occurs when $|\theta - \hat{\theta}| > 1/(2K)$, where the message point is
 $\theta$ and the estimate of the message point after $n$ transmissions is
  $\hat{\theta}$. Here $\hat{\theta}$ is the posterior median after $n$ transmissions.

One problem with this comparison is due to the usage of a
discrete posterior distribution, whereas posterior matching needs a 
continuous posterior distribution for its convergence results. 
Therefore, our simulations--which are actually an approximation 
of the posterior matching scheme--would give FERs that are an upper 
bound on the FERs with a continuous posterior distribution.
Table~\ref{tab:comparison} shows a comparison of the closed-loop 
scheme with an LDPC code depicted in Fig. 2 of~\cite{XiBa07}. 
In the table, the second column shows the error rates of the 
open-loop code from~\cite{XiBa07} and the third column shows 
the error rate of the code used in our paper.  The open-loop code 
has parameters $k=35$, $n=210$ and a rate of 35/210. 
The rate is lower--and hence more conservative--than the rate of our 
closed-loop code, which has rate $\log_2(K)/n = 0.332$, for 
$K = 100000$ and $n = 50$. From the table, we infer that the closed-loop 
scheme does better in terms of FER even when 
we use a discrete approximation and for a much smaller block length.
Hence multi-bit versions
of our scheme would outperform open loop schemes and would also scale
better for applications such as multi-unit auctions and time-varying bid
prices.
{
\begin{table*}
 \centering
 {\small
    \caption{Comparison of Open-Loop and Closed-Loop Schemes}
    \begin{tabular}{| c | c | c |}
    \hline
    $p$ & FER open-loop (code 2 from~\cite{XiBa07}) & FER closed-loop (from our simulations) \\ \hline
    0.05 & $>10^{-1}$ & $0.018$ \\
    0.06 & $>10^{-1} $ & $0.031$ \\
    0.1 & $>10^{-1} $ & $0.049$ \\
    \hline
    \end{tabular}
    \vspace{0.25in}
    \label{tab:comparison}
    }
    \end{table*}
}

\section{Conclusion}
\label{sec:conclusion}
We have presented a realistic micro-level view of auctions in secondary
spectrum markets by explicitly modeling the process by which bidders convey
their bids to a clearing authority. Specifically, we have modeled quantization 
and noise for the first time in this context. In the constant bid
scenario, we have proved that our scheme is optimal in the sense of asymptotically maximizing 
the CA's revenue. We have also 
extended the scheme to accommodate strategic bidders,
to auction multiple spectrum units among strategic bidders and to track 
slowly varying bid prices. 
For the case of strategic users, we develop quantized auction schemes that make
truthful bidding a weakly dominant strategy.
Our simulations verify the theoretical results that we proved in Section~\ref{sec:matched_convergence}. 
Further, the simulations also show the effectiveness of our tracking procedure 
and its robustness to different parameters of the bid-drift model. Our extensions illustrate the importance of low rate feedback since our schemes scale well in both situations, whereas open-loop schemes would have prohibitive overheads.
\vspace*{-0.5cm}
\section{Appendix}
\subsection{Update equations for matched auctioning}
\label{sec:derivation}

We assume that the discretization interval for the bid prices is $\Delta$.
In the continuous case, $\Delta = 0$.  In the discrete case, $\Delta > 0$,
and the bid prices are integer multiples of $\Delta$ with probability 1.

We fix a user $i$ and derive the procedure for computing the posterior 
distribution $F_{i,t}$ from $F_{i,t-1}$ and $y_{i,t}$. To make the notation
lighter, we drop the index $i$ and denote the posteriors by $F_t$ and $F_{t-1}$,
respectively, and the bit received by the CA from the $i$-th user in round $t$
by $y_t$.  We denote by $m_t$ the specific realization of the median of 
$F_{t-1}$, computed by the CA.  We let $m'_t = m_t - \Delta$.

To relate $F_t$ to $F_{t-1}$ and $y_t$, we use the Bayes rule and the total probability
theorem:
\begin{align}
F_t(b) =&
\mathbf{P}\{B\leq b | {\bf Y}_t\}
\nonumber \\
=& \mathbf{P}\{B\leq b | y_t, {\bf Y}_{t-1}\}
\nonumber \\
=& \dfrac{\mathbf{P}\{y_t | B\leq b, {\bf Y}_{t-1}\} \mathbf{P}\{B\leq b | {\bf Y}_{t-1}\}}
         {\mathbf{P}\{y_t | {\bf Y}_{t-1}\}}
\nonumber \\
=& \dfrac{\mathbf{P}\{y_t | B\leq b, {\bf Y}_{t-1}\} F_{t-1}(b)}
         {\mathbf{P}\{y_t | {\bf Y}_{t-1}\}} 
\nonumber \\
=& \mbox{$\frac{\mathbf{P}\{y_t | B\leq b, {\bf Y}_{t-1}\} F_{t-1}(b)}
         {\mathbf{P}\{y_t | B\leq b, {\bf Y}_{t-1}\} F_{t-1}(b) + 
          \mathbf{P}\{y_t | B> b, {\bf Y}_{t-1}\} [1-F_{t-1}(b)]}$.}
\label{eq:Bayes}
\end{align}
We now evaluate the two terms that occur in the denominator of (\ref{eq:Bayes}).
We consider two cases separately: $b<m_t$ and $b\geq m_t$.
\\
{\bf Case 1:} $b<m_t$.\\
In this case, event $B\leq b$ implies that the bid price $B$ is below the median $m_t$,
and therefore $x_t = 0$.  Hence, the first term in the numerator of (\ref{eq:Bayes})
can actually be rewritten as follows:
\begin{equation}
\mathbf{P}\{y_t | B\leq b, {\bf Y}_{t-1}\} = \mathbf{P}\{y_t | x_t = 0, B\leq b, {\bf Y}_{t-1}\}.
\label{eq:aux1}
\end{equation}
Recall that we assume the channel noise to be both temporally independent and independent
of the users' messages.  Therefore, given $x_t=0$, the received bit $y_t$ is conditionally
independent of both the bid price $B$ and all the past received messages ${\bf Y}_{t-1}$.
This means that (\ref{eq:aux1}) can be rewritten as $\mathbf{P}\{ y_t | x_t = 0\}$.
This is equal to the probability of error in BSC$_p$ if $y_t = 1$ and to the probability
of correct reception if $y_t = 0$.  We denote this quantity by $r_t$:
\begin{equation}
\mathbf{P}\{y_t | B\leq b, {\bf Y}_{t-1}\} = py_t + (1-p)(1-y_t) \equiv r_t.
\label{eq:aux2}
\end{equation}
The second term in the denominator of (\ref{eq:Bayes}) is
\begin{align}
& \mathbf{P}\{y_t | B> b, {\bf Y}_{t-1}\} \mathbf{P}\{B>b | {\bf Y}_{t-1} \}
= \mathbf{P}\{y_t, B> b | {\bf Y}_{t-1} \}
\nonumber\\
& = \mathbf{P}\{y_t, b<B<m_t | {\bf Y}_{t-1} \} + \mathbf{P}\{y_t, B\geq m_t | {\bf Y}_{t-1} \}
\nonumber\\
& = \mathbf{P}\{y_t | b<B<m_t, {\bf Y}_{t-1} \}\mathbf{P}\{b<B<m_t | {\bf Y}_{t-1} \}
\nonumber\\
& + \mathbf{P}\{y_t | B\geq m_t, {\bf Y}_{t-1} \}\mathbf{P}\{B\geq m_t | {\bf Y}_{t-1} \}
\nonumber\\
& = r_t \left[F_{t-1}(m'_t) - F_{t-1}(b) \right] + (1 - r_t)\left[1 - F_{t-1}(m'_t) \right]
\nonumber\\
& = -r_tF_{t-1}(b) + (2r_t - 1)F_{t-1}(m'_t)+1-r_t.
\label{eq:aux3}
\end{align}
We now substitute (\ref{eq:aux2}) and (\ref{eq:aux3}) back into (\ref{eq:Bayes})
to obtain the update formula for the case $b<m_t$
\begin{align}
F_{t}(b) =& \dfrac{r_tF_{t-1}(b)}{1-r_t + (2r_t - 1)F_{t-1}(m'_t)}
\nonumber\\
=&
\left\{\begin{array}{ll}
       \dfrac{(1-p)F_{t-1}(b)}{p + (1-2p)F_{t-1}(m'_t)} & \mbox{if } y_t=0,\\ 
       \dfrac{pF_{t-1}(b)}{1-p + (2p - 1)F_{t-1}(m'_t)} & \mbox{if } y_t=1.
       \end{array}\right.
\label{eq:cases13}
\end{align}
\\
{\bf Case 2:} $b\geq m_t$.\\
We proceed similarly to evaluate the two terms in the denominator of (\ref{eq:Bayes}).
We start with the second term:
\begin{align}
\mathbf{P}\{ y_t | B>b, {\bf Y}_{t-1}\} =& \mathbf{P}\{ y_t | x_t = 1 \}
\nonumber\\
 =& (1-p) y_t + p(1-y_t)
\nonumber\\
 =& 1-r_t.
\label{eq:aux4}
\end{align}
For the first term, we have:
\begin{align}
& \mathbf{P}\{y_t | B\leq b, {\bf Y}_{t-1}\} F_{t-1}(b) =
  \mathbf{P}\{y_t, B\leq b | {\bf Y}_{t-1}\} \nonumber \\
& = \mathbf{P}\{y_t, m_t \leq B\leq b | {\bf Y}_{t-1}\} +
  \mathbf{P}\{y_t, B < m_t | {\bf Y}_{t-1}\} \nonumber \\
& = \mathbf{P}\{y_t | m_t \leq B\leq b, {\bf Y}_{t-1}\}\mathbf{P}\{m_t \leq B\leq b | {\bf Y}_{t-1}\}
\nonumber \\
& + \mathbf{P}\{y_t | B < m_t, {\bf Y}_{t-1}\}\mathbf{P}\{B < m_t | {\bf Y}_{t-1}\} \nonumber \\
& = (1-r_t)\left[F_{t-1}(b) - F_{t-1}(m'_t) \right] + r_tF_{t-1}(m'_t) \nonumber \\
& = (1-r_t)F_{t-1}(b) + (2r_t - 1)F_{t-1}(m'_t).
\label{eq:aux5}
\end{align}
Substituting (\ref{eq:aux4}) and (\ref{eq:aux5}) into (\ref{eq:Bayes}), we obtain
the update formula for the case $b\geq m_t$:
\begin{align}
F_t(b) =&
\mbox{$\frac{(1-r_t)F_{t-1}(b) + (2r_t - 1)F_{t-1}(m'_t)}
      {(1-r_t)F_{t-1}(b) + (2r_t - 1)F_{t-1}(m'_t) + (1-r_t)(1-F_{t-1}(b))}$}
\nonumber\\
=&
\dfrac{(1-r_t)F_{t-1}(b) + (2r_t - 1)F_{t-1}(m'_t)}
      {1-r_t + (2r_t - 1)F_{t-1}(m'_t)}
\nonumber\\
=&
\left\{\begin{array}{ll}
       \dfrac{pF_{t-1}(b)+ (1-2p)F_{t-1}(m'_t)}{p + (1-2p)F_{t-1}(m'_t)} & \mbox{if } y_t=0,\\ 
       \dfrac{(1-p)F_{t-1}(b) + (2p - 1)F_{t-1}(m'_t)}{1-p + (2p - 1)F_{t-1}(m'_t)} & \mbox{if } y_t=1.
       \end{array}\right.
\label{eq:cases14}
\end{align}
{\footnotesize
\bibliographystyle{IEEEbib}
{
\bibliography{references}
}
}
\end{document}